\documentclass[letterpaper, 10 pt, conference]{ieeeconf}  

\usepackage[top=57pt,bottom=43pt,left=48pt,right=48pt]{geometry}

\IEEEoverridecommandlockouts                              

\usepackage{amsmath} 
\usepackage{amsfonts}
\usepackage{mathtools} 
\usepackage{amssymb}
\usepackage{graphicx}
\usepackage{bm}  
\usepackage[colorinlistoftodos, textsize=footnotesize]{todonotes}
\usepackage[colorlinks=true, allcolors=black]{hyperref}
\usepackage{setspace}
\usepackage{algorithm}      
\usepackage{algpseudocode}  
\usepackage{tikz}
\usepackage{verbatim}
\usepackage{xcolor}
\usepackage{subcaption}

\usepackage{adjustbox}

\usepackage[normalem]{ulem}

\usepackage[thinc]{esdiff}

\newtheorem{prop}{Proposition}
\newtheorem{lemma}{Lemma}

\usepackage{mathrsfs}
\usepackage{booktabs}

\newif\ifcompile

\DeclareMathOperator{\tr}{tr}

\newcommand{\vf}[1]{{\bm{#1}}}
\newcommand{\mf}[1]{{\mathbf{#1}}}

\newcommand{\bgamma}{{\bm\gamma}}

\newcommand\defeq{\stackrel{\text{\tiny def}}{=}}

\title{\LARGE \bf
A Convex Formulation of Frictional Contact between Rigid and Deformable Bodies
}

\author{Xuchen Han, Joseph Masterjohn, Alejandro Castro
\thanks{All the authors are with the Toyota Research Institute, USA, {\tt\small
firstname.lastname@tri.global}.}%
}

\begin{document}

\maketitle
\thispagestyle{empty}
\pagestyle{empty}

\begin{abstract}
We present a novel convex formulation that models rigid and deformable bodies
coupled through frictional contact. The formulation incorporates a new
corotational material model with positive semi-definite Hessian, which allows us
to extend our previous work on the convex formulation of compliant contact to
model large body deformations. We rigorously characterize our approximations and
present implementation details. With proven global convergence, effective
warm-start, the ability to take large time steps, and specialized sparse
algebra, our method runs robustly at interactive rates. We provide validation
results and performance metrics on challenging simulations relevant to robotics
applications. Our method is made available in the open-source robotics toolkit
Drake.
\end{abstract}

\section{Introduction}

Robotics is on the cusp of a new era. Cutting-edge advancements in machine
learning, computational power, and manufacturing techniques are pushing the limits
of hardware design, sensing, modeling, and control.

The development of soft robotic systems \cite{bib:schmitt2018soft}, deformable
object manipulation capabilities \cite{bib:zhu2022challenges}, and vision-based
tactile sensors \cite{bib:kuppuswamy2020soft, bib:agarwal2021simulation} has
opened up new possibilities for robotics applications in controlled industrial
environments and unstructured home environments. However, these new systems
present challenges in modeling and simulation, particularly when the traditional
rigid-body assumptions of robotic simulators are no longer valid
\cite{bib:zhu2022challenges, bib:lin2021softgym}. One major challenge is the
need for unified simulators that can handle both rigid and deformable objects
that interact through frictional contact \cite{bib:zhu2022challenges}. Despite
these challenges, there is still much potential for growth and innovation in the
field of robotics, and the development of open-source simulation capabilities
targeting robotic applications could play a crucial role in advancing the field.

\begin{figure}[!h]
    \centering
    \adjincludegraphics[width=1.0\columnwidth]{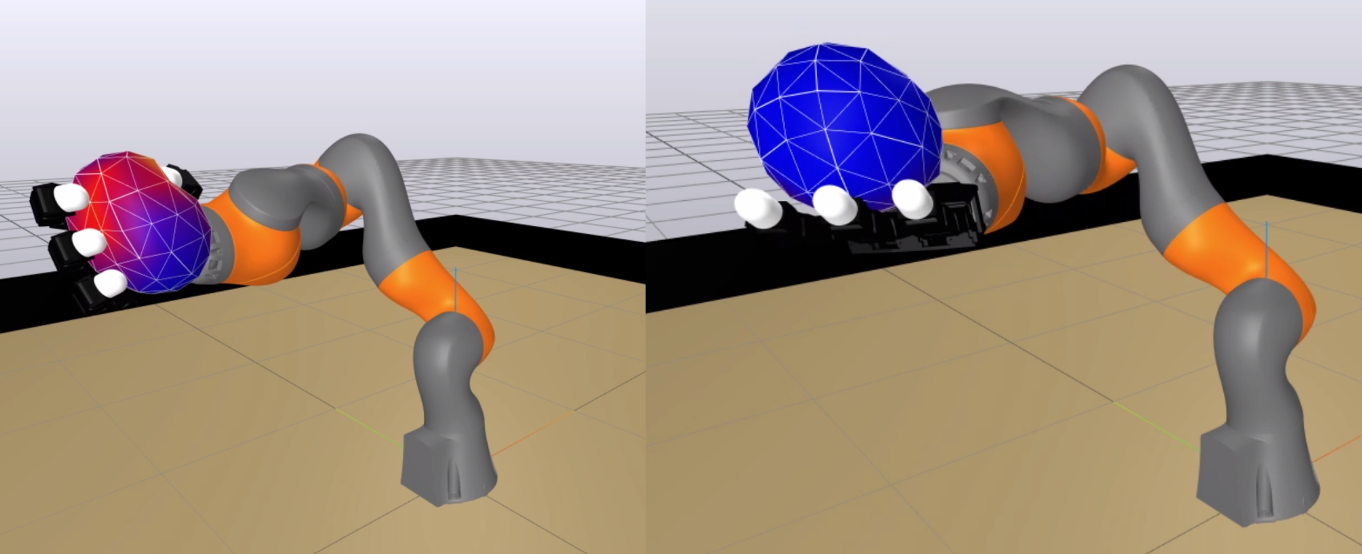}
    \caption{\label{fig:allegro} Simulated robot teleoperated to grasp of a
    deformable ball (see supplemental video). Contours colored by strain
    visualize the state of deformations of the ball. Our method resolves the
    frictional contact between the rigid hand and the deformable ball, enabling
    a secure grasp at interactive rates.}
\end{figure}
\section{Previous Work}

Early work in the finite elements community focused on the mathematical
formalism and model fidelity of methods for the accurate computation of stresses
and deformations at contacting surfaces. Node-on-segment approaches
\cite{bib:hallquist1985sliding, bib:chaudhary1986solution,
bib:bathe1997constraint} enforce contact such that nodes on a contactor surface
do not penetrate their opposing target segments (or facets in three dimensions).
Two-pass approaches iterate this process by reversing the role of
contactor/target surfaces \cite{bib:bittencourt1998}. However, these methods
proved to lack robustness and lead to locking due to over-constraint. To address
these problems, the mortar formalism \cite{bib:MadayOriginalMortarWork1988} was
introduced in segment-to-segment approaches for contact
\cite{bib:bib:mcdevitt2000mortar, bib:puso2004mortar}.

The computer graphics community has contributed significantly, particularly on
performant methods for visually realistic simulations. Projected Gauss Seidel
(PGS) variants \cite{bib:muller2007position, bib:shinar2008two} are popular
given they can be easily parallelized. However, they lack robustness due to
convergence issues \cite{bib:erleben2007velocity}. In the pursuit of increased
robustness, optimization based solvers were developed
\cite{bib:Kaufman2008,bib:gast2015optimization,bib:li2020ipc}. Work on elements
that can recover from inversion \cite{bib:irving2004invertible,
bib:teran2005robust, bib:stomakhin2012energetically} and on corotational models
\cite{bib:mcadams2011efficient, bib:muller2004interactive} further increased the
robustness of these formulations at large deformations.

Multibody dynamics with frictional contact is complicated by the non-smooth
nature of the solutions. Non-existence of solutions \cite{bib:baraff1993issues},
exponential worst-case complexity \cite{bib:baraff1994fast}, and NP-hardness
\cite{bib:Kaufman2008} have led the community to search alternative
formulations. To improve computational tractability, Anitescu introduces a
\textit{convex approximation} of the contact problem \cite{bib:anitescu2006}.
For a strictly convex formulation with unique solution, \cite{bib:todorov2014}
introduces regularization. A convex formulation of compliant contact is
presented in \cite{bib:castro2022unconstrained} and later extended to allow
continuous surface patches in \cite{bib:masterjohn2022velocity}.

Several open-source software options support deformable object modeling. MuJoCo
\cite{bib:todorov2012mujoco} simulates deformable bodies as a group of connected
rigid bodies with spring-dampers. Chrono \cite{bib:chrono2016}  and Siconos
\cite{bib:acary2019siconos} implement the Finite Element Method (FEM), though
mostly target large scale applications, such as granular media simulation. Game
engines like Bullet \cite{bib:bullet} and PhysX \cite{bib:physx} are popular in
reinforcement learning. SOFA \cite{bib:sofa2012}, initially designed for virtual
reality, has also been successful in soft robotics projects.
\section{Outline and Novel Contributions}

To the knowledge of the authors, this work presents the first convex formulation
for the modeling of bodies undergoing large deformations coupled with
articulated rigid bodies through frictional contact and holonomic constraints.

We make a number of contributions. We introduce a novel linear corotational
model of large deformations with a positive semi-definite Hessian in
Section~\ref{sec:corotational_model}. Section~\ref{sec:convex_formulation} uses
this model to extend our previous work on the convex formulation of compliant
contact \cite{bib:castro2022unconstrained} to model deformable bodies.
Sections~\ref{sec:constraint_resolution} and \ref{sec:schur_complement} show how
to exploit sparsity and reuse intermediate computations in our Cholesky
factorizations to write a much smaller contact problem. We present validation
and simulation results relevant to robotics in Section \ref{sec:results} along
with performance metrics. Finally, as part of the open-source robotics toolkit
Drake \cite{bib:drake}, we implement our method with
exhaustive testing and thorough documentation.

\section{Mathematical Formulation}

The state of our system is described by generalized positions $\mf{q} \in
\mathbb{R}^{n_q}$ and generalized velocities $\mf{v} \in \mathbb{R}^{n_v}$ where
$n_q$ and $n_v$ denote the total number of generalized positions and velocities,
respectively. Time derivatives of the configurations are related to the
generalized velocities by $\dot{\mf{q}} = \mf{N}(\mf{q})\mf{v}$, with
$\mf{N}(\mf{q}) \in \mathbb{R}^{n_q \times n_v}$ being the kinematic map. In
particular, we use joint coordinates to describe articulated rigid bodies and
FEM to spatially discretize deformable bodies. As a result, for deformable
bodies, $\mf{q}$ and $\mf{v}$ are the stacked positions and velocities of the
mesh vertices, and the kinematic map $\mf{N}$ is the identity.

\subsection{Kinematics of Constraints}
\label{sec:constraint_kinematics}

We consider $n_c$ constraints described at the velocity level by the constraint
velocity $\mf{v}_{c,i}\in\mathbb{R}^{r_i}$, with $r_i$ the number of equations
for the $i\text{-th}$ constraint. For contact constraints, $\mf{v}_{c,i}$
corresponds to the relative velocity vector at the $i\text{-th}$ contact pair
and $r_i=3$, see \cite{bib:castro2022unconstrained}. For a holonomic constraint
described by $\mf{g}_i(\mf{q},t)=0$, $\mf{v}_{c,i} = d\mf{g}_i/dt$. We can write
constraint velocities in terms of the generalized velocities of the system
(including both rigid and deformable body degrees of freedom) as
$\mf{v}_{c,i}=\mf{J}_i\mf{v} + \vf{b}_i$, with $\mf{J}_i$ and
$\vf{b}_i=\partial\mf{g}_i/\partial t$ the Jacobian and bias, respectively.
Collecting all constraints, we define the stacked constraint velocity
$\mf{v}_{c}=\mf{J}\mf{v} + \mf{b}$, of size $r=\sum r_i$, where $\mf{J}$ and
$\mf{b}$ are the stacked Jacobian and bias term for all constraints.

\subsection{Discrete Time Stepping Scheme}
\label{sec:discrete_time_stepping}

We follow the scheme and notation in
\cite{bib:castro2022unconstrained} closely. We discretize time into intervals of
fixed size $\delta t$ to advance the state of the system from time $t_n$ to the
next step at $t_{n+1} = t_n + \delta t$. To simplify notation, we use the naught
subscript to denote quantities at $t_n$ while no additional subscript is used
for quantities at $t_{n+1}$. We define quantities evaluated at intermediate time
steps $t^\theta = \theta t^{n+1}+(1-\theta)t^{n}$ in accordance with the
standard $\theta\text{-method}$ using scalar parameters $\theta$ and
$\theta_{vq} \in [0, 1]$
\begin{align}
	\mf{q}^{\theta} &\defeq \theta\mf{q} + (1-\theta)\mf{q}_0,\nonumber\\
	\mf{v}^{\theta} &\defeq \theta\mf{v} + (1-\theta)\mf{v}_0,\nonumber\\
	\mf{v}^{\theta_{vq}} &\defeq \theta_{vq}\mf{v} + (1-\theta_{vq})\mf{v}_0,
	\label{eq:theta_method}
\end{align}
so that we can accommodate backward Euler with $\theta=\theta_{vq}=1$,
symplectic Euler with $\theta=0$, $\theta_{vq}=1$  and the second order midpoint
rule with $\theta=\theta_{vq}=1/2$ in the same formulation. Using these
definitions we can write the constrained dynamics of rigid and deformable bodies
within a unified framework as follows
\begin{flalign}
  &\mf{M}(\mf{q}^{\theta})(\mf{v}-\mf{v}_0) = \delta t \,
   \mf{k}(\mf{q}^{\theta},\mf{v}^{\theta}) + \mf{J}(\mf{q}_0)^T\mf{\bgamma},
   \label{eq:scheme_momentum}\\
  &\mathcal{C} \ni \bgamma \perp \mf{v}_c-\hat{\mf{v}}_c \in \mathcal{C}^*,
  \label{eq:convex_constraints}\\
  &\mf{q} = \mf{q}_0 + \delta t \mf{N}(\mf{q}^{\theta})\mf{v}^{\theta_{vq}}.
  \label{eq:scheme_q_update}
\end{flalign}

Equation~\eqref{eq:scheme_momentum} collects the momentum equations for both
rigid and deformable degrees of freedom.  For each articulated rigid body and
deformable body with $n_{v_b}$ degrees of freedom, we use $\mf{M}_b$ to denote
its mass matrix and $\mf{k}_b$ to encode external forces and gyroscopic terms
for articulated rigid bodies \cite{bib:castro2022unconstrained} and internal
forces for deformable bodies. We use $\mf{M}$ to denote the block diagonal
matrix with $\mf{M}_b$ as the diagonal blocks and $\mf{k}$ to denote the stacked
$\mf{k}_b$ for all bodies. Rigid and deformable degrees of
freedom are coupled through the constraint Jacobian $\mf{J}(\mf{q}_0)$ whenever
constraint impulses $\bgamma$ are non-zero. Following \cite{bib:mazhar2014},
Eq.~\eqref{eq:convex_constraints} encodes both holonomic constraints and the
convex approximation of contact constraints. The convex set $\mathcal{C} =
\mathcal{C}_1 \times \mathcal{C}_2 \times \cdots \times \mathcal{C}_{n_c}$ is
the Cartesian product of sets $\mathcal{C}_{i}$ for the $i\text{-th}$
constraint. For contact constraints, $\mathcal{C}_{i}$ corresponds to the
friction cone $\mathcal{F}_{i}$ for that contact
\cite{bib:castro2022unconstrained}. For bi-lateral holonomic constraints,
$\mathcal{C}_{i}=\mathbb{R}^{r_i}$ and impulses can take any real value. The
\emph{dual} of $\mathcal{C}$ is denoted with $\mathcal{C}^*$. For contact
constraints, $\hat{\mf{v}}_{c,i}$ encodes information about the contact distance
at the previous time step, see \cite{bib:castro2022unconstrained} for details.
For holonomic constraints $\hat{\mf{v}}_{c,i} = -\mf{g}_i(\mf{q}_0)/\delta t$,
where $\mf{g}_i(\mf{q}_0)$ is the constraint error at the previous time step
\cite{bib:mazhar2014}. Finally, Eq. \eqref{eq:scheme_q_update} is the positions
update in accordance to the $\theta\text{-method}$.

Next we discuss the modeling of deformable bodies and refer the reader to
\cite{bib:castro2022unconstrained} for details on the treatment of articulated
rigid bodies.

\section{Modeling Deformable Bodies}
\label{sec:deformables_modeling}
Here we describe the equations of motion for a single deformable body and drop
the $b$ subscript of $\mf{M}$ and $\mf{k}$ for simplicity.

\subsection{Preliminaries}
\label{sec:deformables_modeling}

We follow a standard formulation of continuum mechanics discretized by FEM. We
describe it briefly to provide context and introduce notation, and refer the
reader to \cite{bib:bonet2021book, bib:belytschko2014book} for further details.
Within this framework, we describe choices specific to our work.

The kinematics of the model is fully specified by the Lagrangian description map
$\bm{\phi}(\bm{X},t)$ which describes the current position of material points in
the solid as a function of their reference position $\bm{X}$ and time $t$. We
consider hyperelastic solids modeled with an energy density function
$\Psi(\mf{F}(\bm{X}))$, where $\mf{F}=\partial\bm{\phi}/\partial\bm{X}$ is the
deformation gradient. The elastic potential energy for the entire solid is given
by
\begin{equation*}
    E = \int \Psi(\mf{F}(\bm{X})) d\bm{X},
\end{equation*}
which induces elastic force $\mf{f}_e = -\partial E / \partial \mf{q}$. The
stiffness matrix is defined as $\mf{K}=-\partial\mf{f}_e/\partial\mf{q}$. We use
the Rayleigh model of damping
\begin{equation}
    \mf{f}_d(\mf{v}) = -\left(\alpha \mf{M} + \beta \mf{K} \right) \mf{v},
\end{equation}
where $\alpha$ and $\beta$ are Rayleigh damping coefficients and $\mf{M}$ is the
mass matrix from FEM discretization.
With these internal forces, $\mf{k}$ in Eq.~\eqref{eq:scheme_momentum} can be
written as
\begin{equation}
\mf{k}(\mf{v}) = \mf{f}_e(\mf{q^\theta}(\mf{v})) + \mf{f}_d(\mf{v^\theta}(\mf{v})) + \mf{f}_{\text{ext}},
\end{equation}
where $\mf{f_{\text{ext}}}$ includes constant external forces, such as a
gravity.

\subsection{Corotational Material Model}
\label{sec:corotational_model}

In principle, our framework can use any hyperelastic energy density
$\Psi(\mf{F}(\bm{X}))$. We refer to \cite{bib:sifakis2012fem} for implementation
notes and to \cite{bib:kim2020dynamic} for a survey of popular choices. However,
the stiffness matrix $\mf{K}$ can often become indefinite and hinder the
convergence of iterative Newton solvers at large deformations
\cite{bib:teran2005robust, bib:kim2020dynamic}. Standard remedies include
projecting the energy density Hessian to be positive semi-definite
\cite{bib:teran2005robust} or clamping eigenvalues to be non-negative
\cite{bib:smith2019analytic}, both associated with additional computational
cost. On the other hand, linear constitutive models introduce unacceptable
artifacts when large rotational deformation is present
\cite{bib:muller2004interactive}. 

We propose a corotational linear model that enjoys the computational efficiency
of the linear model while avoiding the rotational artifacts at the same time. In
particular, we adopt a Saint-Venant-Kirchhoff model with a corotational linear
Green strain $\hat{\mf{E}}$
\begin{equation}\label{eq:energy_density}
    \Psi(\hat{\mf{E}}) = \mu \|\hat{\mf{E}}\|_F^2 + \frac{\lambda}{2} \tr(\hat{\mf{E}})^2
\end{equation}
where $\mu$ and $\lambda$ are Lam\'{e} parameters.
To define the corotational $\hat{\mf{E}}$, we first define a
\emph{corotated deformation gradient}
\begin{equation*}
\hat{\mf{F}}(\bm{X}) \defeq \hat{\mf{R}}(\bm{X})^T \mf{F}(\bm{X})
\end{equation*}
where $\hat{\mf{R}}(\bm{X})$ is a rotation matrix to be chosen to eliminate the
rotational component of the deformation. Letting $\mf{B} = \hat{\mf{F}} -
\mf{I}$, the Green strain is given by 
\begin{align*}
    \mf{E} &= \frac{1}{2}\left(\mf{F}^T \mf{F} - \mf{I}\right) \\
    &= \frac{1}{2}\left(\hat{\mf{F}}^T \hat{\mf{F}} - \mf{I}\right) = \frac{1}{2}\left(\left(\mf{I} + \mf{B}^T\right) \left(\mf{I} + \mf{B}\right) - \mf{I}\right).
\end{align*}

By dropping the non-linear term $\mf{B}^T \mf{B}$, we arrive at the definition
of the linearized Green strain 
\begin{align*}
\hat{\mf{E}} &\defeq \frac{1}{2}(\hat{\mf{F}} + \hat{\mf{F}}^T) - \mf{I} .
\end{align*}

With $\hat{\mf{R}}=\mf{R}$, the rotational component of $\mf{F}$ from its polar
decomposition, $\hat{\mf{E}}$ becomes the strain measure used in
\cite{bib:mcadams2011efficient}. While this eliminates rotational artifacts,
the non-linear $\Psi$ can lead to indefiniteness in $\mf{K}$.
Instead, we approximate $\hat{\mf{R}}$ with $\mf{R}_0$, the rotational component
of the polar decomposition of $\mf{F}_0(\bm{X})$ at the previous time step. With
this choice, our model is linear, rotational artifacts are negligible (see results in
Section \ref{sec:results}), and the stiffness matrix $\mf{K}$ has an analytic
form guaranteed to be positive semi-definite, as proved in the Appendix. Though
similar to \cite{bib:muller2004interactive}, we do not approximate $\mf{K}$ by
dropping non-linear terms in the derivatives \cite{bib:barbic2012exact}.
Our formulation is linear in $\mf{q}$ by design and naturally generalizes to non-tetrahedral elements.

\section{Convex Formulation of Constrained Dynamics}
\label{sec:convex_formulation}

From Eq.~\eqref{eq:scheme_momentum} we define the momentum residual as
\begin{equation}
	\mf{m}(\mf{v}) = 
	\mf{M}(\mf{q}^{\theta}(\mf{v}))(\mf{v}-\mf{v}_0) -
	\delta t\,\mf{k}(\mf{q}^{\theta}(\mf{v}), \mf{v}^{\theta}(\mf{v})).
	\label{eq:m_definition}
\end{equation}
With our choice of linearized corotational model, the balance of momentum for
deformable bodies is linear in $\mf{v}$. For rigid bodies, we follow
\cite{bib:castro2022unconstrained} where inertia and damping terms are treated
implicitly for stability and Coriolis and gyroscopic forces are treated
explicitly. This allows us to write a linearized version of
Eq.~\eqref{eq:scheme_momentum} as
\begin{equation}
  \mf{m}(\mf{v_0}) + \mf{A}(\mf{v} - \mf{v_0}) = \mf{A}(\mf{v}-\mf{v}^*) = \mf{J}^T\bgamma, \label{eq:momentum_linearized}
\end{equation}
where $\mf{A}=\partial\mf{m}/\partial\mf{v}$ is symmetric positive definite
(SPD) and independent of $\mf{v}$ and we define 
\begin{equation}
	\mf{v}^* \defeq \mf{v_0} - \mf{A}^{-1}\mf{m}(\mf{v_0})
	\label{eq:free_motion_velocities}
\end{equation}
as the \emph{free-motion} velocities that would result in the absence of
constraints, when $\bgamma=\mf{0}$ in Eq.~\eqref{eq:momentum_linearized}. Since
$\mf{A}$ is block diagonal, we solve the free-motion velocities for each body
separately. We discuss efficiency considerations in Section
\ref{sec:schur_complement}. From Section \ref{sec:deformables_modeling},
the block in $\mf{A}$ for the $b\text{-th}$ deformable body is given by
\begin{equation}
    \mf{A}_b = \left(1 + \alpha \theta \delta t\right)\mf{M}_b + \theta \delta t ( \theta_{vq} \delta t + \beta) \mf{K}_b.
    \label{eq:tangent_matrix}
\end{equation}
We note that this result is exact. Moreover $\mf{A}_b$ is SPD since $\mf{M}_b$
is SPD and, with our corotational model, $\mf{K}_b$ is positive semi-definite.
We combine the linearized momentum balance \eqref{eq:momentum_linearized} with
constraints \eqref{eq:convex_constraints} to write a convex formulation of the
dynamics
\begin{align}
\begin{split}
  &\mf{A}(\mf{v}-\mf{v}^*) = \mf{J}^T\bgamma,\\
  &\mathcal{C} \ni \bgamma \perp \mf{v}_c-\hat{\mf{v}}_c \in \mathcal{C}^*.
  \label{eq:optimiality_conditions}
\end{split}
\end{align}

Equations~\eqref{eq:optimiality_conditions} are the optimality conditions to the
following convex optimization problem \cite{bib:mazhar2014,
bib:castro2022unconstrained}
\begin{equation}
	\begin{aligned}
	\min_{\mf{v}} \quad &
	\frac{1}{2}\Vert\mf{v}-\mf{v}^*\Vert_{A}^2 \\
	\textrm{s.t.} \quad & \mf{J}\mf{v} + \mf{b}-\hat{\mf{v}}_c \in \mathcal{C}^*,\\
	\end{aligned}
	\label{eq:sap_primal}
\end{equation}
where we used $\mf{v}_c = \mf{J}\mf{v} + \mf{b}$ from
Section~\ref{sec:constraint_kinematics}. To solve Eq.~\eqref{eq:sap_primal}, we
use the Semi-Analytic Primal (SAP) solver from
\cite{bib:castro2022unconstrained} that formulates a regularized version of this
problem with proven global convergence to its unique solution.

\subsection{Participating Degrees of Freedom}
\label{sec:constraint_resolution}

SAP \cite{bib:castro2022unconstrained} solves a non-linear system of
equations of size $n_v$, usually large with deformable bodies in the system. It
is often impractical to solve such a large scale non-linear system for
interactive robotics applications, so we aim to exploit the structure of the
program \eqref{eq:sap_primal} to reduce its size. 

The constraint Jacobian $\mf{J}$ is often sparse because the constraint velocity
$\mf{v}_{c,i}$ for the $i\text{-th}$ constraint only involves a local set of
vertices. Moreover, since contact develops at the surface of deformable objects,
most columns of $\mf{J}$ that correspond to inner mesh vertices are zero.
Degrees of freedom (DoFs) associated with non-zero columns in $\mf{J}$ are
herein referred to as \emph{participating} DoFs. All other DoFs are referred to
as \emph{non-participating}. For the rest of this section we work on a single
body and drop the $b$-subscript. We use $m_p$ and $m_n$ to denote the number of
participating and non-participating DoFs for this body and in general use
subscripts $p$ and $n$ for participating and non-participating quantities,
respectively. With this partition, we permute DoFs so that non-participating
DoFs appear before participating ones
\begin{gather}\label{eq:permuted_constraint_equation}
  \begin{bmatrix}  
      \mf{A}_{nn} &  \mf{A}_{np} \\
      \mf{A}_{pn}  &  \mf{A}_{pp} 
  \end{bmatrix}
  \begin{bmatrix}  
    \mf{\Delta v}_{n} \\
    \mf{\Delta v}_{p}
  \end{bmatrix}
  =
  \begin{bmatrix}
    \mf{0} \\
    \mf{J}_{p}^T \bgamma 
  \end{bmatrix}
\end{gather}
with $\mf{\Delta v} \defeq \mf{v} - \mf{v^*}$ and $\mf{J}_p$ being the columns
of $\mf{J}$ corresponding to participating DoFs of the deformable body under
consideration. Notice how this partition leads to an (often large) block of
zeros of size $m_n$ on the right hand side. This allows us to eliminate
$\mf{\Delta v}_{n}$ algebraically
\begin{equation}
 \mf{\Delta v}_{n} = -\mf{A}_{nn}^{-1} \mf{A}_{np} \mf{\Delta v}_p
 \label{eq:dv_n}
\end{equation}
to obtain a system in terms of participating DoFs only
\begin{equation}\label{schur_complement_constraint_equation}
  \hat{\mf{A}}\mf{\Delta v}_p = \mf{J}_{p}^T \bgamma 
\end{equation}
where
$\hat{\mf{A}} = \mf{A}_{pp}  -  \mf{A}_{pn} \mf{A}_{nn}^{-1} \mf{A}_{np} \in
\mathbb{R}^{m_p\times m_p}$ is the Schur complement of $\mf{A}$. Since
constraints in \eqref{eq:sap_primal} only involve participating DoFs, we can now
solve a much smaller optimization problem in terms of participating DoFs
\begin{equation}
	\begin{aligned}
	\min_{\mf{v}} \quad &
	\frac{1}{2}\Vert\mf{v}-\mf{v}_{p}^*\Vert_{\hat{A}}^2 \\
	\textrm{s.t.} \quad & \mf{J}\mf{v} + \mf{b}-\hat{\mf{v}}_c \in \mathcal{C}^*,\\
	\end{aligned}
	\label{eq:sap_primal_participating}
\end{equation}
where $\mf{v}_{p}^*$ corresponds to the free-motion velocities of participating
DoFs only. We then use \eqref{eq:dv_n} to update non-participating DoFs
velocities.

\subsection{Schur Complement Computation}
\label{sec:schur_complement}

Explicitly forming $\mf{A}_{nn}^{-1}$ needed in the Schur complement
$\hat{\mf{A}}$ is computationally expensive and impractical. However, we can
obtain $\hat{\mf{A}}$ as an intermediate result during the factorization of
$\mf{A}$ as required for the computation of free-motion velocities in
Eq.~\eqref{eq:free_motion_velocities}. We start from the factorization
of the permuted $\mf{A}$
\begin{equation}
  \begin{bmatrix}  
      \mf{A}_{nn} &  \mf{A}_{np} \\
      \mf{A}_{pn}  &  \mf{A}_{pp} 
  \end{bmatrix}
  =
  \begin{bmatrix}  
      \mf{L}_{nn} &  \mf{0} \\
      \mf{L}_{pn}  &  \mf{L}_{pp} 
  \end{bmatrix}
  \begin{bmatrix}  
      \mf{L}_{nn} &  \mf{0} \\
      \mf{L}_{pn}  &  \mf{L}_{pp} 
  \end{bmatrix}^T
\end{equation}
to see that
\begin{align*}
\hat{\mf{A}} =~& \mf{A}_{pp}  -  \mf{A}_{pn} \mf{A}_{nn}^{-1} \mf{A}_{pp} \\
= ~& \mf{L}_{pn} \mf{L}_{pn}^T +  \mf{L}_{pp} \mf{L}_{pp}^T 
-\mf{L}_{pn} \mf{L}_{nn}^T (\mf{L}_{nn} \mf{L}_{nn}^T)^{-1} \mf{L}_{nn} \mf{L}_{pn}^T \\
= ~& \mf{L}_{pp} \mf{L}_{pp}^T \\
= ~& \mf{A}_{pp} - \mf{L}_{pn} \mf{L}_{pn}^T,
\end{align*}

Therefore, we factorize the permuted $\mf{A}$ with a right-looking Cholesky
factorization and record the Schur complement matrix $\hat{\mf{A}}$ after block
$\mf{A}_{nn}$ has been factorized. Moreover, the same factorization of $\mf{A}$
is used to solve Eqs.~\eqref{eq:free_motion_velocities} and ~\eqref{eq:dv_n}.
Equations within participating and non-participating partitions can be reordered
to reduce fill-ins in the Cholesky factorization. We observe that AMD ordering
\cite{bib:amestoy1996approximate} within a partition leads to fewer fill-ins
than arbitrary ordering.

\section{Results and Discussion}
\label{sec:results}

We present a series of simulation test cases to assess the robustness, accuracy,
and performance of our method. All simulations are carried out in a system with
two 16-core Intel\textsuperscript{\textregistered}
Xeon\textsuperscript{\textregistered} Gold 6226R Processors and 192 GB of RAM,
running Ubuntu 20.04. However, all of our tests run in a single thread. We
report scene statistics and performance numbers in Table \ref{tab:statistics}
and timing breakdowns of the most time-consuming routines in
Fig.~\ref{fig:time_breakdown}. We observe that the narrow phase geometry and the
element-wise FEM computations are embarrassingly parallelizable and can
immediately benefit from a parallel implementation. On the other hand, the
computations for the SAP solver and the Schur complement, both requiring
Cholesky factorizations, are less amenable to parallelization and require
careful investigation. The SAP problem for the arch example essentially is block
tridiagonal, and SAP's supernodal solver \cite{bib:davis2016survey} can
effectively exploit this structure. However, the structure of the problem is
very close to dense for the teddy bear example, and SAP factorizations become
the bottleneck. In addition, in less dynamic scenes (such as the masonry arch
and the cylinder press example), the SAP solver benefits from its warm-start
strategy \cite{bib:castro2022unconstrained}. The scalability of SAP with the
number of bodies is studied in \cite{bib:castro2022unconstrained} and with the
number of constraints in \cite{bib:masterjohn2022velocity}. All examples use our
corotational model except for the cylinder press benchmark, which uses a model
to match published results. Unless otherwise specified, the Rayleigh damping
model uses $\alpha = 0$ and $\beta = 0.01$. For all simulation results, SAP runs
to convergence with relative tolerance $\varepsilon_r=10^{-6}$
\cite{bib:castro2022unconstrained}.

\begin{table*}[t]
  \caption{Timing and scene statistics for all examples. We account for both
   rigid and deformable DoFs, with 3 DoFs per mesh vertex. The realtime rate is
   the ratio of simulated time over CPU time. We also record the maximum number
   of constraints (including contact and holonomic) per time step as well as the
   average number over all time steps.}
  \label{tab:statistics}
  \begin{center}
  \begin{tabular}{l|c|c|c|c} 
    \toprule
    Example                          & DoFs & Time Step (sec) & Realtime Rate & Constraints Average (Max) \\
    \hline
    Arch ($\mu=0.2$)                 & 225  & 0.04            & 7.26                  & 259.5 (357)       \\
    \hline
    Arch ($\mu=1.0$)                 & 225  & 0.04            & 6.29                  & 320.2 (364)       \\
    \hline
    Cylinder (4151 elements)         & 3796 & 0.01            & 0.06                  & 248.3 (335)       \\
    \hline
    Allegro hand                     & 410  & 0.01            & 1.01                  & 89.5 (328)       \\
    \hline
    Teddy bear                       & 2043 & 0.02            & 0.37                  & 236.5 (464)       \\
    \bottomrule
  \end{tabular}
\end{center}
\end{table*}

\begin{figure}[!h]
  \centering
  \adjincludegraphics[width=1.0\columnwidth, trim={{0.07\width} 0.0 {0.07\width} 0.0},clip]{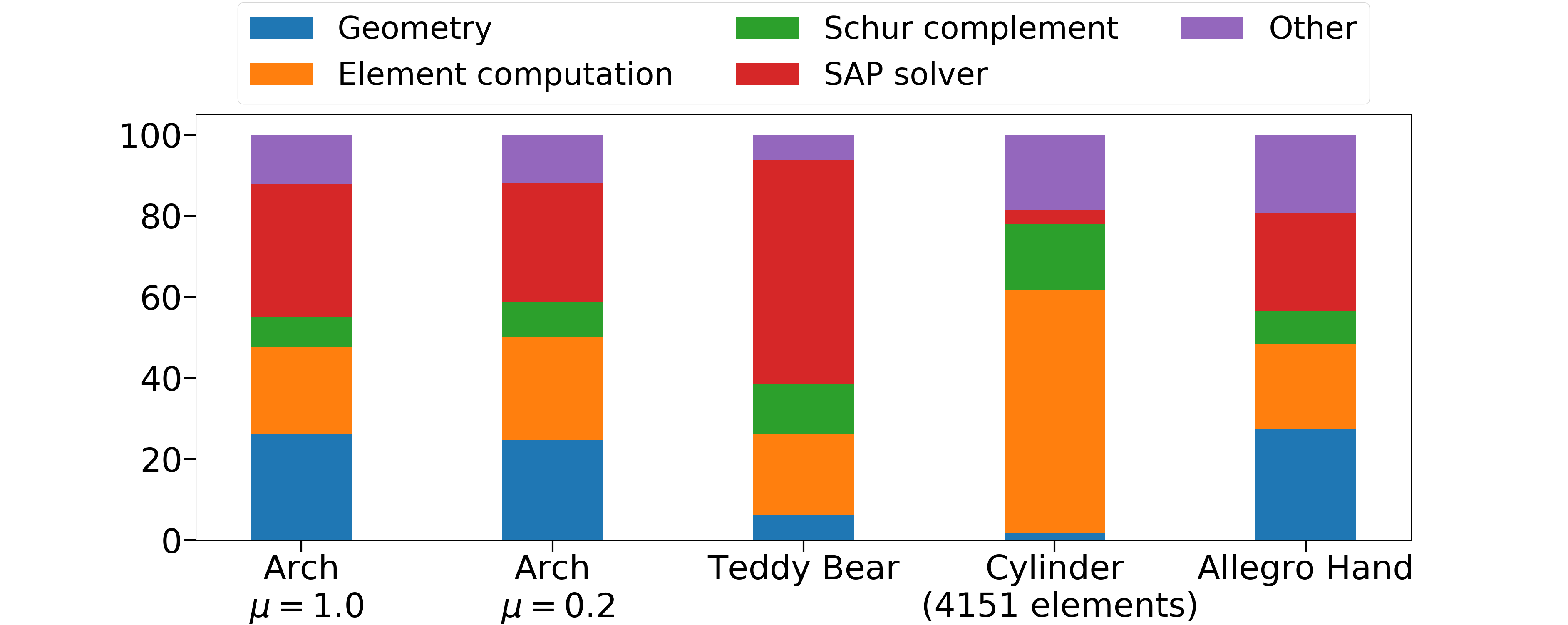}
  \caption{\label{fig:time_breakdown} Cost budget (in percentage) of the total
    run time cost for main computation routines (geometry queries, Schur
    complement, element-wise computation for FEM, and SAP solver) for all
    examples. All other computation are grouped under ``other".}
\end{figure}

\subsection{Rubber Cylinder Pressed Between Two Rigid Plates}
To validate our method in resolving contact as well as the internal stress of
deformable bodies, we perform a benchmark extensively studied in the engineering
literature, with well-known numerical solution \cite{bib:bijelonja2005finite,
bib:sussman1987finite, bib:manual2013vm201}. 

A homogeneous cylinder with a 0.4 meter diameter is pressed between two
frictionless rigid plates. The top plate is pressed downward for a displacement
of 0.2 meter. A schematic of the setup is illustrated in
Fig.~\ref{fig:cylinder_press_schematic}. The material of the cylinder is
modeled as Mooney-Rivlin rubber with constants $C_1 = 0.293~\text{MPa}$, $C_2 =
0.177 ~\text{MPa}$, and bulk modulus $1410~\text{MPa}$. In particular, we use
the Mooney-Rivlin formulation described in \cite{bib:kim2020dynamic} for its
rest stability. The mass density of the material is set to
$1000~\text{kg}/\text{m}^3$ and there is no gravity. As in previous literature,
we consider a plane state of strains. To achieve that in our
inherently 3D formulation, we model the cylinder with two layers of linear
tetrahedral elements along the $z$-direction, assume $\partial/\partial z=0$ in
the computation of the deformation gradient $\bm{F}$, and impose zero out of
plane displacements to avoid drift due to round-off errors. We apply a downward
force on the top plate through a prismatic joint and measure the resulting
displacement. The contact constraints set in as the top plate is pressed
downward and are resolved automatically by our method. We increase the magnitude
of the force linearly in time from zero to $f_{\text{max}} =
1.4~\text{MN}/\text{m}$, to match \cite{bib:bijelonja2005finite}. To better
compare with static analysis results in the literature, we reduce inertial
effects by overdamping the cylinder with $\beta = 0.5$ and by slowly applying
the load over a period of 500 seconds. We perform a mesh refinement study using
four progressively finer meshes. The force-displacement curves obtained with our
method are shown in Fig.~\ref{fig:force_displacement_plot} along with the
reference solution from \cite{bib:bijelonja2005finite}. As the spatial
resolution is increased, the force-displacement curves converge to the reference
solution.

\begin{figure}[!h]
    \centering
    \adjincludegraphics[width=0.8\columnwidth]{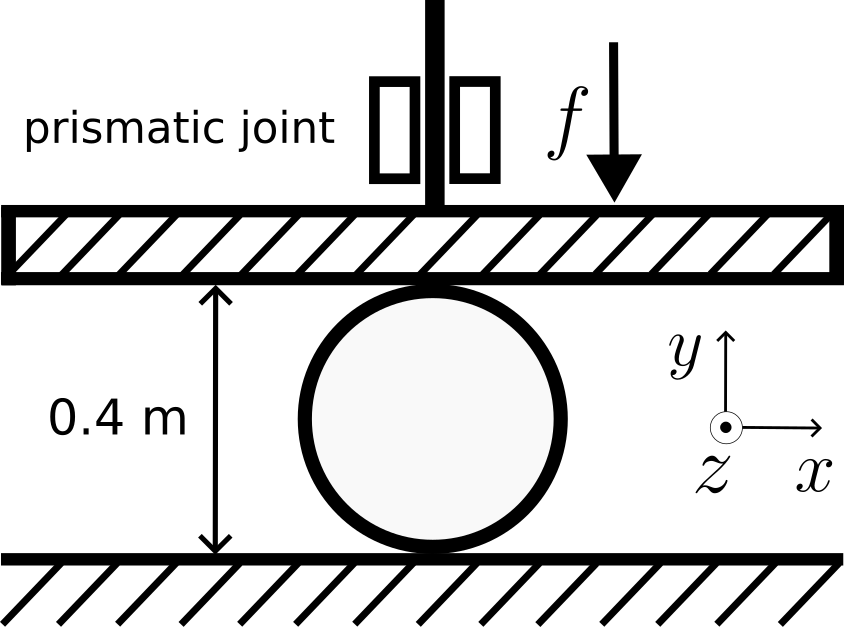}
    \caption{\label{fig:cylinder_press_schematic} A solid deformable cylinder
     0.4 meter in diameter is compressed between a rigid plate and the ground.
     A downward force is applied on the plate through a prismatic joint. Plane
     strain assumption is made in the $xy$-plane.}
\end{figure}

\begin{figure}[!h]
    \centering
    \adjincludegraphics[width=0.8\columnwidth, trim={0 {0.0\width} {0.05\width} {0.04\width}},clip]{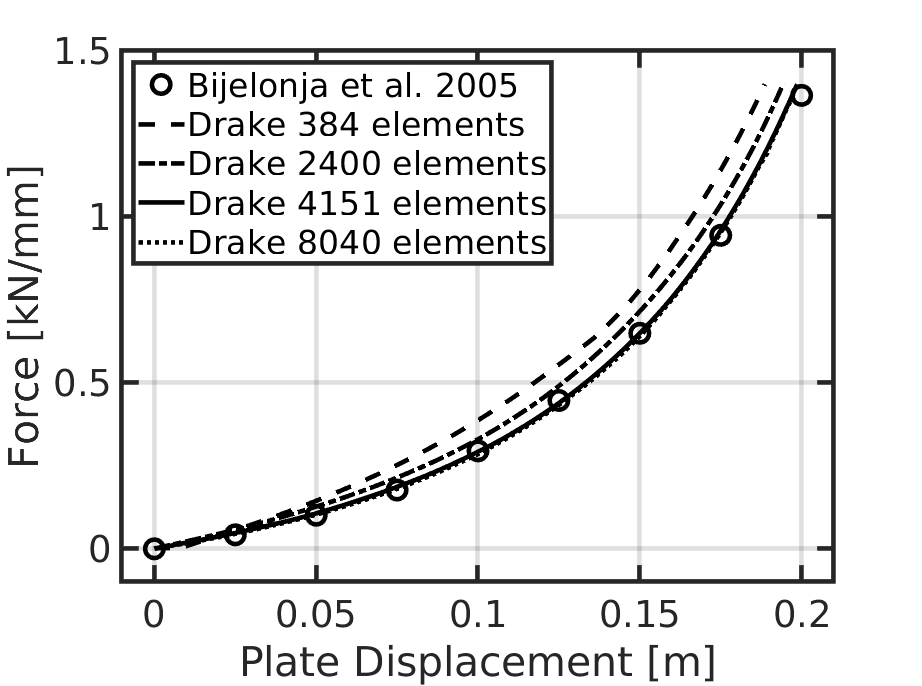}
    \caption{\label{fig:force_displacement_plot} Force per unit length
      vs. displacement for the cylinder press benchmark illustrated in
      Fig.~\ref{fig:cylinder_press_schematic}. The force-displacement curves
      converge under mesh refinement.}
\end{figure}

\subsection{Masonry Arch}
We test the robustness of our method on the frictionally dependent stable arch
tests from \cite{bib:Kaufman2008,bib:li2020ipc,bib:macklin2019nonsmooth}. In
particular, we take the additionally challenging setup from \cite{bib:li2020ipc}
where the arch has its base balanced on sharp edges. Similar to
\cite{bib:li2020ipc}, we simulate the blocks in the arch as very stiff
deformable materials with a high Young's modulus of $20~\text{GPa}$, Poisson's
ratio $0.3$, and density $2300~\text{kg}/\text{m}^3$. While non-linear models
often struggle to converge in presence of such stiff materials due to
ill-conditioning, our method solves the free-motion momentum balance in
Eq.~\eqref{eq:free_motion_velocities} to machine precision with a single linear solve
thanks to our linearized corotated model. Our method is able to stably resolve
the frictional contacts among the blocks at $40~\text{ms}$ time step. We
simulate the structure for 10 minutes and verify the blocks with friction
coefficient $\mu = 1.0$ form a long-term stable arch. We also confirm that the
structure falls apart with a lower friction coefficient $\mu = 0.2$. Our method
has guaranteed convergence and therefore simulation results always satisfy
momentum balance, the model of Coulomb friction, and the maximum dissipation
principle.

\begin{figure}[!h]
    \centering
    \adjincludegraphics[width=1.0\columnwidth]{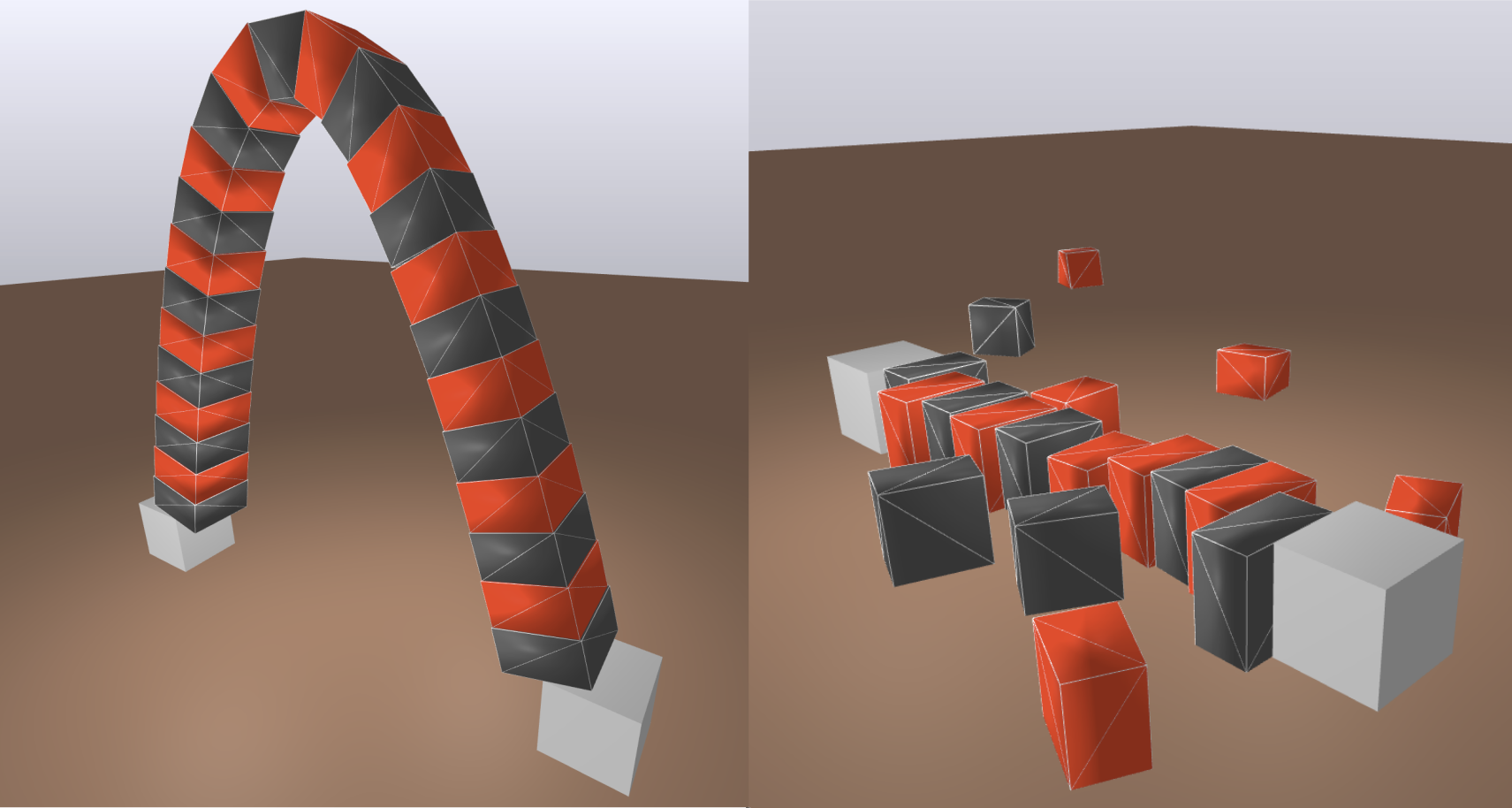}
    \caption{\label{fig:arch} A self-supporting masonry arch. Each block is
     modeled as a stiff deformable object to stress test the stability of our
     method and showcase its ability to resolve contact constraints among
     deformable bodies. With $\mu = 1.0$, the structure is held stable by
     friction (left) whereas with $\mu = 0.2$, the structure eventually falls
     apart (right).}
\end{figure}

\subsection{Soft-bubble Gripper}
\label{sec:bubble_gripper}
We simulate a \textit{Soft-bubble} gripper \cite{bib:kuppuswamy2020soft}
manipulating a deformable teddy bear, Fig. \ref{fig:teddy}. This
is a Schunk WSG 50 gripper with air filled rubber chambers as fingers providing
highly compliant gripping surfaces. We model the air-inflated latex membranes as
deformable volumetric objects with Young's modulus $50~\text{kPa}$ and Poisson's
ratio $0.3$. The flat base of the bubbles are attached to the gripper fingers
using holonomic constraints. The teddy bear is modeled with Young's modulus
$10~\text{kPa}$ and Poisson's ratio $0.4$. The  gripper is PD-controlled with a
prescribed \textit{close-lift-shake-place} motion sequence. We perform the
shaking motion to stress test the stability of the grasp in simulation, see the
supplemental video. The entire system with 2043 degrees of freedom is highly
coupled via up to 464 constraints. Our solver is able to solve all time steps of
this challenging scenario to convergence.
\begin{figure}[!h]
    \centering
    \adjincludegraphics[width=0.8\columnwidth]{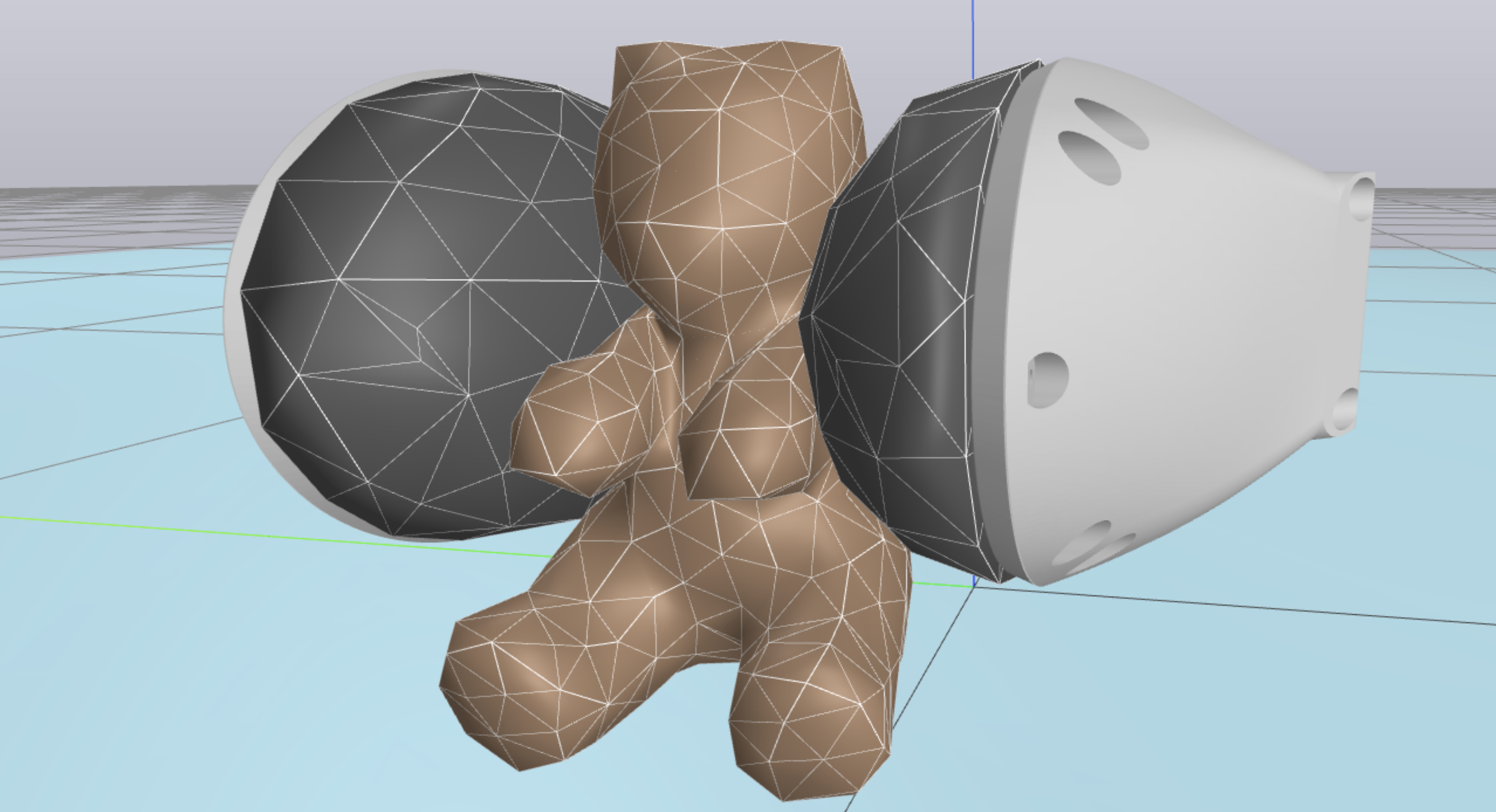}
    \caption{Highly compliant \label{fig:teddy} \emph{Soft-bubble} gripper
    \cite{bib:kuppuswamy2020soft} securing the grasp of a deformable teddy
    bear.}
\end{figure}

\subsection{Interactive Control}
\label{sec:allegro_hand}
A KUKA LBR iiwa arm (7 DoFs) outfitted with an anthropomorphic Allegro hand (16
DoFs) is teleoperated in real-time to manipulate a deformable ball, Fig.
\ref{fig:allegro}. Desired joint positions and velocities are computed using a
differential inverse kinematics controller that targets a desired end effector
pose, controlled interactively with a gamepad. The deformable ball is modeled
with Young’s modulus $25~\text{kPa}$ and Poisson’s ratio $0.4$. As the robot
closes its palm around the ball, contact patches between the ball and the palm,
phalanges, and fingertips of the hand form a secure grasp. We apply a shaking
motion to stress test the stability of the grasp, see supplemental video. Even
with the overhead of the controller, our simulation runs at real-time rates.

\section{Limitations and Future Work}

We summarize the limitations we have identified for our method and propose
future research directions.

\textbf{Parallel implementation:} 
With a focus on accuracy and algorithms, our implementation for this work is
serial. Some of the most time-consuming routines in our method can easily
benefit from a parallel implementation, while the same is not obvious for the
SAP solver and the Schur complement computation. Leveraging the power of
parallelization on modern hardware for these computations is an interesting area
for future investigation.

\textbf{Rotational invariance:} 
As with all other linear constitutive models, our linearized model with lagged
rotational component is not rotationally invariant. Thus it is not suitable for
simulation of extreme deformations using large time steps. For those scenarios,
we fall back to traditional nonlinear models with Hessian positive definite
corrections proposed in \cite{bib:teran2005robust}.

\textbf{Self-contact:} 
We do not consider self-contact at the moment due to the lack of support by our
geometry engine. Self-contact can be incorporated into our method by updating the
geometry engine to augment the set of contacts reported.

\textbf{Tunneling at high speeds:} Though our method has a lower computational
cost, it could benefit from continuous collision detection strategies
\cite{bib:li2020ipc} to provide constraints before contact is established. This
would allow to mitigate issues such as objects tunneling past each other at high
speeds. Efficient solution to mitigate this issue is a topic of active research
for the authors.

\textbf{Redundant constraints:} Our geometry engine often introduces a large
number of constraints to resolve contact. Similarly, welding a large number of
deformable mesh vertices to a rigid body (as done in Section
\ref{sec:bubble_gripper}) introduces many constraints. Even though our SAP
solver \cite{bib:castro2022unconstrained} provides existence and uniqueness
guarantees, a large number of constraints hurts performance as can be observed
in the \emph{Soft-bubble} example. We are currently investigating strategies to
significantly reduce the number of constraints without sacrificing accuracy.

\section{Conclusions}

We present what we believe is the first convex formulation of bodies undergoing
large deformations coupled with articulated rigid bodies through contact and
holonomic constraints.

To achieve this, we introduce a material model in terms of a linearized Green
strain and demonstrate that our formulation is linear with a positive
semi-definite stiffness matrix when the rotational component of the deformation
gradients is lagged. This allows us to incorporate the modeling of large
deformations into our previous work on the convex formulation of contact
\cite{bib:castro2022unconstrained}. We exploit the structure of the problem in
two ways. Firstly, we partition the problem and express it as a smaller
optimization program that only includes constrained variables. Secondly, we show
how the expensive-to-compute Schur complements required in the optimization
problem can be obtained as an intermediate computation in the Cholesky
factorization of the momentum balance equations.

To demonstrate the effectiveness of our method, we present validation results
and simulations relevant to robotics. Specifically, we showcase the ability of
our approach to resolve stable grasps in two challenging manipulation tasks
robustly. 

Unlike other approaches that run a fixed number of iterations to stay within a
specified computational budget, our method always solves the fully constrained
problem to convergence at real-time rates. This is enabled by the strong
convergence guarantees of our method using the SAP solver
\cite{bib:kuppuswamy2020soft}. Therefore, our simulation results always satisfy
the momentum equations and friction laws without introducing difficult-to-detect
artifacts. This aspect is particularly critical for the meaningful
sim-to-real transfer of results.

Profiling reveals that parallelization of geometry and FEM elemental routines is
attractive in cases with a much larger number of degrees of freedom than
constraints. However, cases dominated by the number of constraints are less
amenable to parallelization since they are dominated by the cost of
factorizations and the SAP solver. Our analysis includes a summary of
limitations and future research directions.

Finally, we incorporated our method into the open-source robotics toolkit Drake
\cite{bib:drake} and hope that the simulation and robotics communities can
benefit from our contribution.

\section*{Appendix}
\label{sec:appendix}

\begin{lemma}\label{lemma1}
  For $\mf{M} \in \mathbb{R}^{n\times n}$,
  \begin{equation*}
    \tr(\mf{M}^2) + \tr(\mf{M}^T \mf{M}) \ge 0.
  \end{equation*}
\end{lemma}
\begin{proof}
  Let $\mf{H}$ and $\mf{S}$ be the symmetric and skew-symmetric part of $\mf{M}$ respectively. Then
  the desired sum is given by
  \begin{equation*}
  \tr((\mf{M} + \mf{M}^T) \mf{M}) = 2\tr(\mf{H}(\mf{H}+\mf{S})) = 2\tr(\mf{H}^2) \ge 0,
  \end{equation*}
  where the last equality uses the fact that the trace of the product of a symmetric
  matrix and a skew-symmetric matrix is zero.
\end{proof}

\begin{prop}
  The consitutive model defined in Eq.~\eqref{eq:energy_density} is linear in $\mf{q}$
  and the resulting stiffness matrix $\mf{K}$ is positive semi-definite.
\end{prop}
\begin{proof}
  Recall that $\Psi(\hat{\mf{E}}) = \mu \|\hat{\mf{E}}\|_F^2 + \frac{\lambda}{2} \tr(\hat{\mf{E}})^2$ and that $\hat{\mf{E}} = \frac{1}{2}(\hat{\mf{R}}^T\mf{F} + \mf{F}^T\hat{\mf{R}}) - \mf{I}$ for some constant $\hat{\mf{R}}$.
  Differentiating $\Psi$ with respect to $\mf{F}$ twice gives
  \begin{equation*}
   \frac{\partial^2 \Psi}{\partial F_{ij} \partial F_{kl}} = \mu\left(\delta_{ik} \delta_{jl} + \hat{R}_{il}\hat{R}_{kj}\right) + \lambda \hat{R}_{ij}\hat{R}_{kl},
  \end{equation*}
  which is constant since $\hat{\mf{R}}$ is constant and thus the model in Eq.~\eqref{eq:energy_density} is linear.
  For an arbitrary $\mf{dF} \in \mathbb{R}^{3\times3}$, $\mf{dF} : \frac{\partial^2 \Psi}{\partial \mf{F}^2}: \mf{dF}$ is given by
  \begin{align*}
    &\mu\left( dF_{ij} dF_{ij} + dF_{kl}\hat{R}_{il}dF_{ij}\hat{R}_{kj}\right)+ \lambda dF_{ij}\hat{R}_{ij}dF_{kl}\hat{R}_{kl} \\
    &= \mu\left( \tr(\mf{dF}^T\mf{dF}) + \tr(\mf{dF}\hat{\mf{R}}^T\mf{dF}\hat{\mf{R}}^T)\right) + \lambda\tr^2(\mf{dF}\hat{\mf{R}}^T)\\
    &\ge 0.
  \end{align*}
  The $\lambda$ term is clearly non-negative and the $\mu$ term is non-negative by applying Lemma \ref{lemma1} with $\mf{M} = \mf{dF}\hat{\mf{R}}^T$.
  Consequently, for an arbitrary $\mf{dq} \in \mathbb{R}^{n_v}$,
  \begin{align*}
    \mf{dq}^T~\mf{K}~\mf{dq} =&~\mf{dq}^T~\frac{\partial^2 E}{\partial \mf{q}^2}~\mf{dq} \\
                           =& \int \left(\frac{\partial \mf{F}}{\partial \mf{q}} \mf{dq}\right) : \frac{\partial^2 \Psi}{\partial \mf{F}^2} : \left(\frac{\partial \mf{F}}{\partial \mf{q}} \mf{dq}\right) d\bm{X}
                           \ge 0,
  \end{align*}
  and thus $\mf{K}$ is positive semi-definite.
\end{proof}





\section*{Acknowledgment}
We thank the Dynamics \& Simulation and Dexterous Manipulation teams at
TRI for their continuous patience and support. We thank Frank Permenter for
helpful discussion and Damrong Guoy for geometry support.

\bibliographystyle{IEEEtran} 
\bibliography{root}

\end{document}